\documentclass{article}
\usepackage{fullpage}
\usepackage{amsmath}
\usepackage{graphicx}

\newtheorem{theorem}{Theorem}
\newtheorem{lemma}[theorem]{Lemma}
\newcommand{\qed}{\hfill\rule{7pt}{7pt}}
\newenvironment{proof}{\noindent{\bf Proof:}}{\qed\medskip}

\title{A method for hedging in Continuous Time}
\author{Yoav Freund}
\begin{document}

\maketitle

\newcommand{\Vmax}{V^M}
\section{Introduction}
This article gives an analysis of the NormalHedge algorithm in
continuous time. The NormalHedge algorithm is described and analyzed
in discrete time in~\cite{ChaudhuriFrHs09a}. The continuous time
analysis is mathematically cleaner, simpler and tighter than the
discrete time analysis.

To motivate the continuous time framework consider the problem of
portfolio management. Suppose we are managing $N$ different financial
instruments allowed to define a desired distribution
of our wealth among the instruments. We ignore the details of the buy and
sell orders that have to be placed in order to reach the desired
distribution, we also ignore issues that have to do with transaction
costs, buy-sell spreads and the like. We assume that at each moment
the buy and sell prices for a unit of a particular instrument are the
same and that there are no transaction costs.

Our goal is to find an algorithm for managing the portfolio
distribution. In other words, we are looking for a mapping from past
prices to a distribution over the instruments. As we are considering
continuous time, the past can be arbitrarily close to the
present. Formally speaking, we say that the portfolio distribution 
is ``causal'' or ``unanticipating'' to remove the possibly of
defining a portfolio which is a function of the future gains as his
would clearly be a cheat. We are interested in considering continuous
time, because instrument prices can fluctuate very rapidly.

To model this very rapid fluctuation we use a type of stochastic
process called an It\^{o} process to model the log of the price as a
function of time. Intuitively, an It\^{o} process is a linear
combination of a differentiable process and white noise. A more formal
definition is given below. To read more about It\^{o} processes
see~\cite{Oskendal03}. 

Our algorithm and its analysis do not make {\em any} additional
assumption on the price movement of the instruments. Of course, with
no additional assumption we cannot have any guarantees regarding our
future wealth. For example, if the price of all of the instruments
decreases at a particular moment by 10\%, our wealth will necessary
decrease by 10\%, regardless of our wealth distribution. However,
surprisingly enough, we {\em can} give a guarantee on the {\em regret}
associated with our method without any additional assumptions.  Regret
quantifies the difference between our wealth at time $t$ and the
wealth we would have had if we invested all of our money in the best
one of the $N$ instruments. Specifically, denote the log price of
instrument $i$ at time $t$ by $X^i_t$ and assume that the initial unit
price for all instruments is one, i.e. $X^i_0=\log(1)=0$. 
Let $G_t$ be the log of our wealth at time $t$. We define our regret
at time t as
\[
R_t = \max_{i=1,\ldots,N} X^i_t - G_t~.
\]
Intuitively, the regret is large if by investing all of our money in a
particular instrument (whose identity is known only in hind-sight) we would
have made much more money than what we actually have. The main result
of this paper is an algorithm for which the regret is bounded by
$\sqrt{2ct(\ln N + 1)}$ where $c$ is the amount of random fluctuations
(white noise) in the instrument prices. It stands to reason that the bound
depends on the number of instruments $N$, because the larger the
potential impact of price fluctuations. If the instrument prices are
are all simple random walks (brownian motion in continuous time) the expected
price of the {\em best instrument} is proportional to $\sqrt{ct\ln N}$.

We can get potentially tighter bounds if we consider the {\em set} of best
instruments. Suppose we sort $X^i_t$ for a particular time $t$ from
the largest to the smallest. We say that the $\epsilon-quantile$ of
the prices is the value $x$ such that $\lfloor \epsilon N \rfloor$ of
$X^i_t$ are larger than $x$. We prove a bound of $\sqrt{2ct(\ln
  (1/\epsilon) + 1)}$ on the regret of our algorithm relative to the
$\epsilon$-quantile for any $\epsilon>0$. This generalized bound can
be used when we are hedging over an infinite, even uncountably
infinite set of instruments.

Case in point. Suppose that the instruments that we are combining are
themselves portfolios. The set of fixed portfolios over $N>1$ instruments
consists of the $N-1$ dimensional simplex, which is an uncountably
infinite set. Cover's universal porfolios
algorithm~\cite{Cover91,CoverOr96} uses this set as the set of
portfolios to be combined. We can apply our algorithm to this set and
guarantee that our regret relative to the top $\epsilon$-quantile of
fixed rebalanced portfolios is small.

\section{Hedging in Continuous time}
The portfolio management problem is defined as follows.  Let $X^i_t$
for $i=1,\ldots,N$ define the log-prices of $N$ instrumentss as a
function of time. The initial unit price is one, thus
$X^i_0=0$. $X^i_t$ is an It\^{o} process. More specifically, let
$dW_t$ denote an $N$ dimensional Wiener process where each coordinate
is an independent process with unit variance. Then the differential of
the total gain corresponding to the $i$th instrument is:
\begin{equation} \label{eqn:dXi-as-Ito}
dX^i_t = \hat{a}^i(t)dt + \sum_{j=1}^N \hat{b}^{i,j}(t) dW^j_t
\end{equation}
Where $\hat{a}^i(t),\hat{b}^{i,j}(t)$ are adapted (non-anticipatory)
stochastic processes that are It\^{o}-integrable with respect to $W^j_t$.
In mathematical finance terms $\hat{a}_i$ correspond to {\em price drift} 
and $\hat{b}^{i,j}$ correspond to diffusion or {\em price volatility}. 

The volatility of the $i$th instrument at time $t$ is defined as
\begin{equation} \label{eqn:Volatility}
\hat{V}^i(t) = \sum_{j=1}^N \left( \hat{b}^{i,j}(t) \right)^2
\end{equation}
and the maximal volitility at time $t$ is defined as
\begin{equation} \label{eqn:Vmax}
\Vmax(t) = \max_i \hat{V}^i(t)
\end{equation}

An ``aggregating strategy'' is an portfolio management policy that defines how
to distribute the wealth among the $N$ instruments as a function of their
past performance. Mathematically speaking, a policy is a stochastic
process that is an adapted (non-anticipative) function of the past
prices ${X_t^i}_{i=1}^N$. Given an instantiation of the stochastic
processes ${X_t^i}_{i=1}^N$ the aggregating strategy defines $N$
stochastic processes $P_t^i$ such that for all $t$ $P_t^i\geq 0$ and
$\sum_i P_t^i=1$. The cumulative gain of the master algorithm is
defined to be $0$ at $t=0$, for $t \geq 0$ it is defined by the
differential
\[
dG_t = \sum_{j=1}^N P^i_t dX_t^j
\]
The regret of the master algorithm relative to the $i$th instrument is
defined to be zero at $t=0$ and is otherwise defined by the
differential
\[
dR^i_t = dX^i_t-dG_t
\]
And we can combine the last two equations to get:
\begin{equation} \label{eqn:sum_R_i}
\sum_{i=1}^N P^i_t dR^i_t = 0
\end{equation}
As $R_t^i$ is a linear combination of $X_t^i$ it is also an It\^{o}
process and can be expressed as
\begin{equation} \label{eqn:dRi-as-Ito}
dR^i_t = a^i(t)dt + \sum_{j=1}^N b^{i,j}(t) dW^j_t
\end{equation}
where 
\begin{equation}
a^i(t)=\hat{a}^i(t) - \sum_{k=1}^N P^k_t \hat{a}^k(t)
\end{equation}
and
\begin{equation} \label{eqn:bij}
b^{i,j}(t) = \hat{b}^{i,j}(t) - \sum_{k=1}^N P^k_t \hat{b}^{k,j}(t)
\end{equation}

Similarly to $X^i_t$ we define the diffusion rate of $R^i_t$ to be
\begin{equation} \label{eqn:Vi}
V^i(t) = \sum_{j=1}^N \left( b^{i,j}(t) \right)^2
\end{equation}

We prove an upper bound on $V^i(t)$
\begin{lemma}
\[
\forall t, V^i(t) \leq 2\Vmax(t)
\]
\end{lemma}

\newcommand{\vbh}{{\bf \hat{b}}}
\newcommand{\vb}{{\bf b}}
\begin{proof}
  We use $\vb^j$ and $\vbh^j$ the $N$ dimensional vectors
  $\left\langle b^{1,j},\ldots,b^{N,j}\right\rangle$ and $\left\langle
    \hat{b}^{1,j},\ldots,\hat{b}^{N,j}\right\rangle$
  respectively.Using this notation we rewrite
  Equations~(\ref{eqn:bij}) and~(\ref{eqn:Vi}) as
\[
\vb^i(t) = \vbh^i(t) - \sum_{k=1}^N P^k_t \vbh^k(t);\;\; 
V^i(t) = \|\vb^i(t) \|_2^2
\]
Equations~(\ref{eqn:Vmax}) and~(\ref{eqn:Vi}) imply that $\|\vbh^i(t)\|_2^2
\leq \Vmax(t)$ for all $i$. It follows that the norm of the convex
combination is also bounded:
\[
\left\| \sum_{k=1}^N P^k_t \vbh^k(t) \right\|_2^2 =
\sum_{k=1}^N P^k_t \left\| \vbh^k(t) \right\|_2^2 \leq
\Vmax(t)
\]
From which it follows that 
\[
\left\|  \vbh^i(t) - \sum_{k=1}^N P^k_t \vbh^k(t) \right\|_2^2 \leq 2 \Vmax(t)
\]
\end{proof}

\section{Normalhedge}

NormalHedge is a particular aggregating strategy which is defined as
follows.

We define a potential function that depends on two variables, $x$ and $c$:
$$
\phi(x,c) = \left\{ \begin{array}{cc} \exp\left( \frac{x^2}{2c} \right) &
(x > 0) \\ 1 & (x \leq 0) \end{array} \right.
$$

We will use the following partial derivatives of $\phi(x,c)$:
$$
\phi'(x,c) \doteq {\partial \over \partial x} \phi(x,c) = \left\{ \begin{array}{cc} \frac{x}{c} \exp\left(
\frac{x^2}{2c} \right) & (x > 0) \\ 0 & (x \leq 0) \end{array} \right.
\quad
\phi''(x,c) \doteq {\partial^2 \over \partial x^2}\phi(x,c) = \left\{ \begin{array}{cc} \left( \frac1{c} +
\frac{x^2}{c^2} \right) \exp\left( \frac{x^2}{2c} \right) & (x > 0) \\ 0 &
(x < 0) \end{array} \right.
$$
and
\[
\phi^c(x,c) \doteq
{\partial \over \partial c} \phi(x,c)  = \left\{ 
\begin{array}{cc} -\frac{x^2}{c^2} \exp\left(\frac{x^2}{2c} \right) & (x > 0) \\
0 & (x \leq 0) \end{array} \right.
\]

The NormalHedge strategy is defined by the following conditions that
should hold for every $t\geq 0$.  If $R^i_t \leq 0$ for all $1 \leq i
\leq N$ then $P^i_t=1/N$. Otherwise $P^i_t$ and $c(t)$ are defined by
the following equations.
\begin{equation} \label{eqn:c(t)}
\frac{1}{N}\sum_{i=1}^N \phi(R^i_t,c(t)) = e
\end{equation}
\begin{equation}
P^i_t = \frac{\phi'(R^i_t,c(t)}{\sum_{j=1}^N \phi'(R^j_t,c(t)}
\end{equation}

\section{Analysis}

\newcommand{\Rg}[2]{R_{{#1},{#2}}}
\newcommand{\eplus}{E_{t}}

We introduce a new notion of regret. For a given time $t$ we order the
cumulative gains $X^i_t$ for $i=1,\ldots,N$ from highest to lowest and
define the {\em regret of the agregating strategy to the top
  $\epsilon$-quantile} to be the difference between $G(t)$ and the
$\lfloor \epsilon N \rfloor$-th element in the sorted list.

\begin{lemma} \label{lem:scale2regret}
At any time $t$, the regret to the best instrument can be bounded as:
\[ \max_{i} \Rg{i}{t} \leq \sqrt{2 c(t) ( \ln N + 1)}  \]
Moreover, for any $0 \leq \epsilon \leq 1$ and any $t$, the regret to the
top $\epsilon$-quantile of instruments is at most
\[ \sqrt{2 c(t) (\ln (1/\epsilon) + 1)}. \]
\end{lemma}

\begin{proof}
The first part of the lemma follows from the fact that,
for any $i \in \eplus$,
\begin{eqnarray*}
\exp\left( \frac{(\Rg{i}{t})^2}{2c(t)} \right)
\ = \
\exp\left(\frac{([\Rg{i}{t}]_+)^2}{2 c(t)}\right) 
\leq
\sum_{i'=1}^N \exp\left( \frac{([\Rg{i'}{t}]_+)^2}{2 c(t)} \right)
\ \leq \ Ne
\end{eqnarray*}
which implies $\Rg{i}{t} \leq \sqrt{2c(t)(\ln N + 1)}$.

For the second part of the lemma, let $\Rg{i}{t}$ denote the regret of
our algorithm to the instrument with the $\epsilon N$-th highest price
at time $t$. Then, the total potential of instruments with regrets
greater than or equal to $\Rg{i}{t}$ is at least:
\[ \epsilon N \exp\left( \frac{([\Rg{i}{t}]_+)^2}{2 c(t)} \right) \leq N e \] 
from which the second part of the lemma follows.
\end{proof}

We quote It\^{o}'s formula, as stated in~\cite{Oskendal03} (Theorem 4.2.1) 
\begin{theorem}[It\^{o}]
Let
\[
dX(t) = udt+vdB(t)
\]
be an n-dimensional It\^{o} process. Let
$g(t,x)=(g_1(t,x),\ldots,g_p(t,x))$ be a $C^2$ map from $[0,\infty)
\times R^n$ into $R^p$. The the process
\[
Y(t,\omega) = g(t,X(t))
\] 
is again an It\^{o} process, whose component number $k$, $Y_k$, is
given by
\[
dY_k=
\frac{\partial g_k}{\partial t}(t,X)dt+
\sum_{i} \frac{\partial g_k}{\partial x_i}(t,X)dX_i+
\frac{1}{2}
\sum_{i,j} \frac{\partial^2 g_k}{\partial x_i \partial x_j}(t,X)dX_i dX_j
\]
where $dB_i dB_j = \delta_{i,j}dt, dB_idt=dtdBi=0.$
\end{theorem}

We now give the main theorem, which characterizes the rate of increase
of $c(t)$.
\begin{theorem}
With probability one with respect to the Weiner process
\[
\forall t,\;\;\; \frac{dc(t)}{dt} \leq 6 \Vmax(t)
\]
\end{theorem}

\begin{proof}
  We denote the potential corresponding to the $i$th instrument by potential
  by $\Phi^i_t$, i.e.
\[
\Phi_t^i = \phi(R_t^i,c(t))
\]
Using It\^{o}'s formula we can derive an equation for the differential
$d\Phi^i_t$:
\begin{eqnarray}
d\Phi^i_t & = &
\left[
\frac{dc(t)}{dt}\phi^c(R^i_t,c(t))
+
a^i(t) \phi'(R^i_t,c(t))
+
\frac{1}{2}\left( \sum_{j=1}^N (b^{i,j}(t))^2 \right) \phi''(R^i_t,c(t))
\right] dt
+
\sum_{j=1}^N b^{i,j}(t) \phi'(R^i_t,c(t)) dW^j_t \nonumber \\
&=& 
\left[
\frac{dc(t)}{dt}\phi^c(R^i_t,c(t))
+
\frac{1}{2}\left( \sum_{j=1}^N (b^{i,j}(t))^2 \right) \phi''(R^i_t,c(t))
\right] dt
+ dR^i_t\phi'(R^i_t,c(t)) \label{eqn:phi_i}
\end{eqnarray}

We sum Equation~(\ref{eqn:phi_i}) over all instruments. As $c(t)$ is chosen so that
the average potential is constant, the differential of the average
potential is zero. We thus get:
\begin{eqnarray}
0 &=& \sum_{i=1}^N d\Phi^i_t \\
&=&  \sum_{i=1}^N \left[
\frac{dc(t)}{dt}\phi^c(R^i_t,c(t))
+
\frac{1}{2}\left( \sum_{j=1}^N (b^{i,j}(t))^2 \right) \phi''(R^i_t,c(t))
\right] dt
+
 \sum_{i=1}^N dR^i_t\phi'(R^i_t,c(t))
\end{eqnarray}
From Equation~(\ref{eqn:sum_R_i}) we know that the last term is equal
to zero. Removing this term and reorganizing the equation we arrive at
an expression for the rate of change of $c(t)$:
\[
\frac{dc(t)}{dt} = - \frac
{\sum_{i=1}^N \left( \sum_{j=1}^N (b^{i,j}(t))^2 \right) \phi''(R^i_t,c(t))}
{2\sum_{i=1}^N \phi^c(R^i_t,c(t))}
\]

we plug in the definitions of $V^i(t)$,$\phi^c$ and $\phi'$ to get:
\[
\frac{dc(t)}{dt} = \frac
{\sum_{i; R_t^i>0} V^i(t) \left( \frac1{c(t)} + \frac{(R^i_t)^2}{c(t)^2} \right) \exp\left( \frac{(R^i_t)^2}{2c(t)} \right)  }
{2\sum_{i; R_t^i>0} \frac{(R^i_t)^2}{c(t)^2} \exp\left(\frac{(R^i_t)^2}{2c(t)} \right) }
\]
Multiplying the enumerator and denominator by $c(t)$, using the bound
$V^i(t) \leq \Vmax$ and denoting $x_i \doteq R^i_t/\sqrt{c(t)}$ we get
the inequality
\begin{equation} \label{eqn:final}
\frac{dc(t)}{dt} \leq
\Vmax(t)
\frac{\sum_{i; x_i>0} (1+x_i^2) e^{x_i^2/2}}
{\sum_{i; x_i>0} x_i^2 e^{x_i^2/2}}
\end{equation}
The maximum of the ratio on the right hand side under the constraint 
$(1/N) \sum_{i; x_i>0} e^{x_i^2/2} = e$ is achieved when
$x_i=\sqrt{2}$ for all $i$. Plugging this value back into
equation~(ref{eqn:final}) yields the statement of the theorem.
\end{proof}

\section{references}
There are many good sources for stochastic differential equations and
the It\^{o} calculus. One which I found particularly appealing is a set of
lecture notes for a course on ``Stochastic Calculus, Filtering, and
Stochastic Control'' by Ramon van Handel, available from the web here:

\begin{tt}
http://www.princeton.edu/$\sim$rvan/acm217/ACM217.pdf
\end{tt}

\bibliography{bib} \bibliographystyle{alpha}

\end{document}